\documentclass[10pt]{article}
\usepackage[utf8]{inputenc}
\usepackage[T1]{fontenc}
\usepackage[margin=1in]{geometry}
\usepackage{amsmath,amssymb,amsthm,booktabs,enumitem,float,multirow,graphicx,microtype}
\newtheorem{proposition}{Proposition}
\usepackage[round]{natbib}
\usepackage{authblk}

\title{\vspace{-2em}\bfseries Bias-robust causal inference for panel data}
\author{Angelos Alexopoulos}
\affil{Department of Economics, Athens University of Economics and Business,
Greece\\ \texttt{angelos@aueb.gr}}
\date{}

\begin{document}
\maketitle
\vspace{-3em}

\begin{abstract}
\noindent
We develop a bias-robust framework for causal inference in panel data models with interactive fixed effects. Conventional imputation-based estimators of average treatment effects on the treated do not explicitly account for counterfactual approximation error, potentially leading to misleading inference. Building on bias-aware minimax methods, we construct an estimator that corrects imputed counterfactuals using weighted untreated residuals and delivers confidence intervals with an explicit allowance for the remaining approximation bias. The resulting framework separates protection against approximation bias from the efficiency gains obtained through optimal correction. In simulations, the correction reduces confidence interval length by approximately 13\% relative to bias-aware inference without correction, while preserving empirical coverage. An application to Election Day Registration demonstrates how the correction strengthens the robustness of substantive conclusions to counterfactual error. The proposed approach provides a transparent sensitivity analysis for causal effects when the accuracy of counterfactual imputation cannot be established from observed data alone.
\end{abstract}

\section{Introduction}

Most policy settings do not offer random assignment, so researchers recover
missing counterfactuals from observational data. Synthetic control matches a
treated unit's pre-treatment path with a weighted combination of controls
\citep{abadie2010}; interactive-fixed-effects models instead use unit-specific
loadings on common factors \citep{bai2009}, and generalized synthetic control
(GSC) and matrix completion exploit that low-rank structure across many treated
units and adoption dates \citep{gobillon2016,xu2017,samartsidis2020,athey2021}.

This flexibility has a cost. Imputation-based factor methods learn factors from control outcomes and each
treated unit's loading from its pre-treatment history, so counterfactual error
enters the treatment effect directly, and conventional standard errors do not account for it. Related concerns motivate synthetic difference-in-differences
\citep{arkhangelsky2021}, augmented synthetic control \citep{benmichael2021},
and generalized-Bayes robustness \citep{alexopoulosdemiris2025}. A separate strand of the literature addresses the parallel trends assumption: \citet{rambachan2023} study robustness to violations of it, and \citet{callaway2021} develop identification and inference under
conditional parallel-trends assumptions; neither addresses robustness to
counterfactual imputation error in factor-model panels.

The contribution of the present paper is to make causal inference bias-robust,
using the bias-aware minimax methods of \citet{armstrong2022}, developed for
regression coefficients in a complete panel. In particular,
\citet[Remark 2.2]{armstrong2022} discuss extending their framework to heterogeneous treatment effects and they note that obtaining a computable bound in that setting is nontrivial. Our target is of that kind. Namely, it is the average treatment effect on the treated (ATT), and it differs from a regression coefficient in a way that matters for the bias analysis. A coefficient is estimated from cells that are all observed; the ATT is defined through untreated potential outcomes that are missing precisely where treatment occurs. The counterfactual error therefore lives in a block of the panel that is never observed untreated, and the residuals available to correct for it come only from cells outside that block. More precisely, we keep the GSC
estimand but choose weights on untreated residuals minimizing worst-case bias
plus variance, and report intervals that allow explicitly for the bias that
remains. Inference rests on an assumed bound for the counterfactual error.
Coverage follows from that bound under high-level conditions, but we do not
derive the bound itself for this design; as in honest
difference-in-differences \citep{rambachan2023} we therefore report it rather
than estimate it, and summarize a conclusion by the largest error it survives.
That is less a concession than a change of accounting. Honest difference-in-differences made the counterfactual error of parallel-trends designs explicit and reportable; imputation-based factor methods also require that error to be small, and simply never say how small.

Section 2 develops the estimator, its bias bound and its interval; Section 3
reports the simulations and Section 4 the real data application. Finally,
Section 5 concludes.

\section{Method}

For $T\times N$ matrices let
$\langle B,G\rangle=\operatorname{tr}(B'G)=\sum_{t,i}B_{ti}G_{ti}$ be the
entrywise inner product, which treats the two matrices as long vectors and takes
their dot product. Let $s_1(B)\geq s_2(B)\geq\cdots$ be the singular values of
$B$, which record how much of it lies along each of its orthogonal directions. Three summaries of
those values appear below: the operator norm $\|B\|_{\mathrm{op}}=s_1(B)$, the
Frobenius norm $\|B\|_F=(\sum_j s_j(B)^2)^{1/2}$ and the nuclear norm
$\|B\|_*=\sum_j s_j(B)$.
With $D_{ti}$ the treatment indicator, $\mathcal O=\{D_{ti}=0\}$ holds all
untreated observations, including the pre-adoption periods of units treated
later, and $\mathcal T=\{D_{ti}=1\}$ the rest.
Outcomes obey
\begin{equation}
Y_{ti}=\Gamma_{ti}+\tau_{ti}D_{ti}+\varepsilon_{ti},\qquad
\Gamma_{ti}=x_{ti}'\beta+\alpha_i+\delta_t+L_{ti},\quad
\operatorname{rank}(L)\leq R .
\label{eq:model}
\end{equation}
Here $\Gamma_{ti}=\operatorname{E}\{Y_{ti}(0)\mid
x_{ti},\alpha_i,\delta_t,L_{ti}\}$ and $\tau_{ti}=Y_{ti}(1)-Y_{ti}(0)$; we
assume consistency, no interference and no anticipation, and
$\operatorname{E}(\varepsilon\mid D,x,\alpha,\delta,L)=0$, so treatment may
depend on the latent untreated-outcome structure, and in particular on the
loadings, but not on the idiosyncratic errors.
Effects may differ across units and periods; with $W=D/|\mathcal T|$ the target
$\tau_W=\langle W,\tau\rangle$ averages them over all treated observations, the first
departure from the regression-coefficient problem. We fit \eqref{eq:model} by
rank-$R$ least squares on $\mathcal O$ alone and complete $\widehat\Gamma$
everywhere. For $A$ supported on $\mathcal O$, that is with $A_{\mathcal T}=0$, let
$\widehat\tau(A)=\langle W,Y-\widehat\Gamma\rangle-\langle
A,Y-\widehat\Gamma\rangle$: the first term is the usual imputation estimate, the
average gap between observed and imputed outcomes on treated observations, and
the second subtracts a weighted average of residuals from untreated
observations, chosen
below to remove as much of the counterfactual error as possible. Write
$E=\Gamma-\widehat\Gamma$ for the counterfactual error matrix and $H=W-A$ for
the matrix that carries it into the estimate, so that
\begin{equation}
\widehat\tau(A)-\tau_W=\langle H,E\rangle+\langle H,\varepsilon\rangle .
\label{eq:decomp}
\end{equation}

\subsection{Worst-case bias and minimax weights}

We motivate the choice of weights through a fixed-weight
bias--variance criterion. For a nonrandom matrix $A$ and a
fixed approximation-error matrix $E$, define
\[
\operatorname{MSE}(A,E)
=
\operatorname{E}\bigl[
\{\widehat\tau(A)-\tau_W\}^2
\bigr],
\]
where the expectation is taken over the idiosyncratic
errors, treating $A$ and $E$ as fixed. This criterion
provides the motivation for the feasible minimax
optimization developed below. However, it cannot be minimized as it stands, since it involves the
counterfactual bias, the first term of \eqref{eq:decomp}, and that term depends
on the error matrix $E$, which is not observed. Following
\citet{armstrong2022} we therefore replace the criterion by the largest value
it can attain when the counterfactual error is restricted to a class
$\mathcal E(C)$, the radius $C$ of which is the magnitude of error that the
analyst is prepared to assume. The maximum over that class involves the weights
alone and no longer involves $E$, so that the resulting minimax problem is well
posed; it is this problem which the present subsection develops. For the class $\mathcal E(C)$ we take, as those authors do, the nuclear ball of
radius $C$ about the origin, $\mathcal E(C)=\{E:\|E\|_*\leq C\}$. The reason
for this choice is that the worst-case bias over such a class is available in
closed form: the nuclear and the operator norm being dual to one another,
\begin{equation}
\sup_{E\in\mathcal E(C)}|\langle H,E\rangle|=C\,\|H\|_{\mathrm{op}},
\label{eq:duality}
\end{equation}
and the supremum is attained when $E$ is of rank one, that is when the whole
counterfactual error consists of a single time path scaled by a single set of
unit loadings.

We now turn to the second term of \eqref{eq:decomp}. The matrix $A$ vanishes on
$\mathcal T$ and the matrix $W$ vanishes on $\mathcal O$, so that the two have
disjoint supports and $\|H\|_F^2=\|W\|_F^2+\|A\|_F^2$. Assuming the errors to
be uncorrelated with common variance $\sigma^2$, we therefore have that
\begin{equation}
\operatorname{Var}\langle H,\varepsilon\rangle=\sigma^2\|H\|_F^2
=\sigma^2\bigl(\|W\|_F^2+\|A\|_F^2\bigr),
\label{eq:var}
\end{equation}
of which only the second term depends on the weights. The correction is thus
able to reduce the bias, but it can do so only at the cost of a variance of its
own. We note also that $A$ is restricted to $\mathcal O$ because untreated
residuals contain counterfactual approximation error and noise, but no
treatment effect; at a treated observation the residual
$Y_{ti}-\widehat\Gamma_{ti}$ contains the very effect which is being estimated. Let $\widehat\sigma^2=|\mathcal O|^{-1}\sum_{(t,i)\in\mathcal O}
(Y_{ti}-\widehat\Gamma_{ti})^2$ denote the mean squared residual over the
untreated observations. Collecting \eqref{eq:duality} and \eqref{eq:var}, the
weights solve
\begin{equation}
\widehat A
=
\arg\min_{A:\ A_{\mathcal T}=0}
\bigl\{
C^2\|W-A\|_{\mathrm{op}}^2
+
\widehat\sigma^2\|A\|_F^2
\bigr\}.
\label{eq:minimax}
\end{equation}
For fixed weights and approximation error, the
bias--variance decomposition in \eqref{eq:decomp},
together with \eqref{eq:duality} and \eqref{eq:var},
motivates the worst-case criterion
$C^2\|W-A\|_{\mathrm{op}}^2+
\sigma^2\|W-A\|_F^2$.
Because $W$ and $A$ have disjoint supports,
$\|W-A\|_F^2=\|W\|_F^2+\|A\|_F^2$.
The term $\sigma^2\|W\|_F^2=
\sigma^2/|\mathcal T|$ does not depend on $A$
and may therefore be omitted from the
optimization. Equation~\eqref{eq:minimax}
implements the resulting minimax criterion
with $\sigma^2$ replaced by the residual-based
estimate $\widehat\sigma^2$. It is thus a
feasible plug-in optimization, rather than
an exact minimization of the unconditional
finite-sample MSE of the data-dependent
estimator. Since the squared operator norm
is convex and the support restriction is
linear, the optimization problem is convex.

Three properties of the solution are worth recording. First, the objective is
homogeneous, so that $\widehat A$ depends on $C$ and $\widehat\sigma$ only
through their ratio; at $C=0$ it returns $\widehat A=0$, which is the
uncorrected imputation estimate, and as $C$ grows so does the correction, but
only until weights confined to $\mathcal O$ are no longer able to flatten the
leading singular direction of $W$ any further, beyond which the estimate
becomes insensitive to $C$. Second, the two terms of \eqref{eq:minimax} need
not be comparable in magnitude. Once $C\|H\|_{\mathrm{op}}$ dominates
$\widehat\sigma\|H\|_F$ the problem is in effect the minimization of an
operator norm subject to the support restriction, and the variance term serves
only to break ties; it matters when the assumed radius is small relative to the
noise, which is also the situation in which the correction has most to gain.
\citet[Remark 2.3]{armstrong2022} observe that the mean squared error is one
criterion among several and that in their simulations the particular choice
makes little difference. It makes almost none here, and for the reason just
given: where the bias allowance dominates, minimizing the worst-case mean
squared error and minimizing the length of the interval are nearly the same
problem, so that replacing the one criterion by the other leaves the intervals
almost unchanged.

The third property concerns the design rather than the data, and determines the
settings in which the correction is of any use at all. For an ever-treated unit
let $g_i=\min\{t:D_{ti}=1\}$ denote its adoption date, and let
$\mathcal Q=\{i:\exists\,t,\ D_{ti}=1\}$ collect the ever-treated units and
$\mathcal P=\{t:\exists\,i,\ D_{ti}=1\}$ the periods in which any unit is
treated, so that $\mathcal P=\{t\geq\min_{i\in\mathcal Q}g_i\}$ when treatment
is absorbing. Let $u_1$ and $v_1$ denote leading singular vectors of $W$, let
$P_{\mathcal O}$ be the operator which sets to zero every entry outside
$\mathcal O$, and define
\begin{equation}
c=\bigl\|P_{\mathcal O}\bigl(u_1v_1'\bigr)\bigr\|_F^2 .
\label{eq:cstat}
\end{equation}
Since $\|u_1v_1'\|_F=1$, the quantity $c$ lies in $[0,1]$ and is the share of
the leading rank-one direction of $W$ which falls on observed cells. It is
determined by the assignment pattern alone and may therefore be computed before
any outcome is used.

\begin{proposition}[Adoption timing and the minimax correction]
\label{prop:active}
Let $C>0$ and $\widehat\sigma^2>0$, let $\mathcal T\neq\emptyset$, and let
treatment be absorbing. The following are equivalent:
\begin{enumerate}[label=(\roman*),nosep]
\item the adoption dates $\{g_i\}_{i\in\mathcal Q}$ are not all equal;
\item $c>0$;
\item $\widehat A\neq0$;
\item the minimum in \eqref{eq:minimax} is strictly below
      $C^2\|W\|_{\mathrm{op}}^2$.
\end{enumerate}
When they fail, $\widehat\tau(\widehat A)$ is the uncorrected imputation
estimate with bias bound $C\|W\|_{\mathrm{op}}$.
\end{proposition}

\begin{proof}
The matrix $D$ is nonnegative, and since treatment is absorbing the
earliest-adopting unit is treated in every period of $\mathcal P$, so that
$(WW')_{ts}>0$ for all $t,s\in\mathcal P$, while every ever-treated unit is
treated in the final period, so that $(W'W)_{ij}>0$ for all
$i,j\in\mathcal Q$. Both Gram matrices are therefore entrywise positive on
$\mathcal P\times\mathcal P$ and on $\mathcal Q\times\mathcal Q$ and vanish
elsewhere, and the Perron--Frobenius theorem gives that $s_1(W)$ is simple with
$u_1$ and $v_1$ strictly positive on $\mathcal P$ and on $\mathcal Q$
respectively and zero off them. Consequently $(u_1v_1')_{ti}\neq0$ precisely on
$\mathcal P\times\mathcal Q$, so that $c>0$ if and only if
$\mathcal O\cap(\mathcal P\times\mathcal Q)$ is nonempty. Treatment being
absorbing, that intersection is
$\{(t,i):i\in\mathcal Q,\ \min_j g_j\leq t<g_i\}$, which is nonempty exactly
when some $g_i$ exceeds $\min_j g_j$. This establishes (i) $\Leftrightarrow$
(ii).

Suppose next that $c>0$ and set $\Delta=P_{\mathcal O}(u_1v_1')$, which
satisfies $\Delta_{\mathcal T}=0$ and is therefore admissible, and for which
$\langle\Delta,u_1v_1'\rangle=\|\Delta\|_F^2=c$. Simplicity of $s_1(W)$ makes
the operator norm differentiable at $W$ with gradient $u_1v_1'$, so that
$\|W-t\Delta\|_{\mathrm{op}}=\|W\|_{\mathrm{op}}-tc+o(t)$ and
\[
C^2\|W-t\Delta\|_{\mathrm{op}}^2+\widehat\sigma^2t^2\|\Delta\|_F^2
=C^2\|W\|_{\mathrm{op}}^2-2C^2\|W\|_{\mathrm{op}}c\,t+o(t),
\]
which lies strictly below $C^2\|W\|_{\mathrm{op}}^2$ for all sufficiently small
$t>0$. Hence (ii) $\Rightarrow$ (iv), and (iv) $\Rightarrow$ (iii) is
immediate. For the converse, $\|W-A\|_{\mathrm{op}}^2$ is convex in $A$ and
$\widehat\sigma^2\|A\|_F^2$ is strictly convex when $\widehat\sigma^2>0$, so
the objective is strictly convex and its minimizer unique; a nonzero
$\widehat A$ cannot then attain the value taken at $A=0$, which gives (iii)
$\Rightarrow$ (iv).

Suppose finally that the adoption dates are all equal to some $g$. Then
$\mathcal T=\{t\geq g\}\times\mathcal Q$ is a product set, and since $A$
vanishes on $\mathcal T$ the restriction of $W-A$ to it is $W$ itself. Taking a
set of rows together with a set of columns cannot increase an operator norm, so
that $\|W-A\|_{\mathrm{op}}\geq\|W\|_{\mathrm{op}}$ for every admissible $A$,
with equality at $A=0$. The first term of \eqref{eq:minimax} is therefore
minimized at $A=0$, though not there alone, and it is $\widehat\sigma^2>0$
which selects that point, the second term being nonnegative and vanishing only
at the origin. Hence $\widehat A=0$.
\end{proof}

The correction thus draws its power from staggered adoption and from nothing
else: it is active precisely when some ever-treated unit is still untreated in
a period during which another is already treated. A single treated unit, or
several units adopting at a common date, leaves no such observation, however
long the panel. The requirement that $\widehat\sigma^2>0$ is not decorative.
Under common adoption with $\widehat\sigma^2=0$ the minimizer is not unique,
since any admissible $A$ supported on rows outside $\mathcal P$ and columns
outside $\mathcal Q$ with $\|A\|_{\mathrm{op}}\leq\|W\|_{\mathrm{op}}$ leaves
$W-A$ block diagonal and its operator norm unchanged. With a rank-$R$ fit on
$\mathcal O$ and continuously distributed outcomes the condition holds almost
surely.

The statistic $c$ grades the condition rather than merely signing it. Among
admissible directions of unit Frobenius norm, $P_{\mathcal O}(u_1v_1')/\sqrt c$
maximizes the first-order decrease, so it is the steepest descent direction and
not merely one which happens to work. Along the unnormalized direction $\Delta$
defined above,
\[
\frac{d}{dt}\Bigl\{C^2\|W-t\Delta\|_{\mathrm{op}}^2
+\widehat\sigma^2t^2\|\Delta\|_F^2\Bigr\}\Big|_{t=0}
=-2C^2\|W\|_{\mathrm{op}}\,c .
\]
How much of that local improvement is realized depends in addition on the
variance penalty and on the gap $s_1(W)-s_2(W)$, which bounds how far the
leading direction can be flattened before a second one binds.

It remains to construct the interval. Let $\widehat H=W-\widehat A$. Evaluating
\eqref{eq:duality} at $\widehat H$ bounds the counterfactual bias by
$b=C\|\widehat H\|_{\mathrm{op}}$, and we take as the scale of the noise the
plug-in quantity $\widehat{\mathrm{se}}=\widehat\sigma\|\widehat H\|_F$, which
is \eqref{eq:var} evaluated at the realized weights with $\widehat\sigma$ in
place of $\sigma$. Equation \eqref{eq:var} is a fixed-weight calculation, so
this is again a feasible plug-in rather than the exact standard deviation of
$\langle\widehat H,\varepsilon\rangle$: $\widehat H$ is data-dependent, and the
adequacy of $\widehat{\mathrm{se}}$ for the implemented procedure is maintained
through the Gaussian approximation of Proposition~\ref{prop:coverage} rather
than established by \eqref{eq:var} itself. We report
$\widehat\tau(\widehat A)\pm(b+z_{1-\alpha/2}\widehat{\mathrm{se}})$, the
rationale being that the estimate lies off centre by at most $b$ and scatters
about that point with standard deviation $\widehat{\mathrm{se}}$, so that
allowing for the bias in full and for the noise at the usual quantile covers
the estimand whatever the true bias may be. A sharper constant is available,
since a bias of known magnitude shifts the sampling distribution rather than
widening it \citep{armstrongkolesar2018}; the saving is negligible in the
designs considered here, because $m=b/\widehat{\mathrm{se}}$ is large, and we
therefore report the conservative interval, which is also the one
\citet{armstrong2022} implement.

Coverage then rests upon two conditions: that $C$ does indeed bound the
counterfactual error, and that the noise is approximately normal on the scale
of $\widehat{\mathrm{se}}$. Granting both at a prespecified $C$ yields a
pointwise asymptotic guarantee at that radius. It is not a simultaneous
statement over the sweep, nor does it cover a radius chosen after seeing which
conclusion survives; what makes the sensitivity analysis of Section 4
meaningful is that the guarantee is available at each radius the analyst is
willing to name in advance. The following is proved in Section S4 of the
supplementary material.

\begin{proposition}[High-level sensitivity coverage]\label{prop:coverage}
Fix $C>0$ and suppose $\Pr(\|E\|_*\leq C)\to1$ and
$\langle\widehat H,\varepsilon\rangle/\widehat{\mathrm{se}}
\overset{d}{\to}N(0,1)$. Then
\[
\liminf\Pr\bigl\{\tau_W\in\widehat\tau(\widehat A)
\pm(b+z_{1-\alpha/2}\widehat{\mathrm{se}})\bigr\}\geq1-\alpha .
\]
\end{proposition}

The proposition adapts the high-level coverage argument of
\citet{armstrong2022}. It is a conditional validity result: granting its two
conditions it delivers coverage, but it does not establish that the
residual-based scale consistently studentizes the estimator as implemented.
We maintain rather than verify its
Gaussian-approximation and standard-error conditions for the factor estimator
which we implement; Section 3 illustrates the finite-sample coverage which
results under the designs stated there, and does not establish the conditions
in general. We note that the homoskedastic form assumed in \eqref{eq:var}
enters the construction twice, in the criterion, where its misspecification
costs efficiency alone, and in $\widehat{\mathrm{se}}$, where it bears upon
coverage. Since the proposition is stated in terms of
$\widehat{\mathrm{se}}$, the conclusion also applies with a dependence-robust
scale estimator in its place, provided the corresponding studentized Gaussian
approximation holds for that estimator.

There remains the choice of $C$. Under their complete-panel assumptions, and
for the preliminary estimator which they specify, \citet{armstrong2022} bound
the relevant nuclear-norm error by $3R\,s_1(\widehat U)(1+\epsilon)$, where
$\widehat U$ denotes the residual matrix; the multiplier is $3R$ when the
target is the error $\Gamma-\widehat\Gamma$ and $2R$ when the target is
redefined, the latter being the version which they implement. Their theorem is
proved for an estimator which refits after an initial debiasing pass, and does
not reach a single rank-$R$ fit on $\mathcal O$ in a panel whose treated block
is never observed untreated. We therefore treat
$\widehat C=3R\,s_1(\widehat U)(1+\epsilon)$ as a diagnostic reference rather
than as a radius which is valid for the present design, and Section 3 reports
how it compares with the realized error there. What Proposition
\ref{prop:coverage} leaves to be supplied is the radius itself, and we
accordingly report it as a sensitivity parameter and summarize a conclusion by
the largest value of $C$ which that conclusion survives.
\section{Simulations}

To assess the performance of the proposed intervals we conducted a simulation
study in which the untreated outcomes contain a factor structure and the
assignment of treatment is correlated with it. We simulate
\[
Y_{ti}=1.5f_{1t}\lambda_{1i}+\kappa f_{2t}\lambda_{2i}+\tau D_{ti}
+\varepsilon_{ti},
\]
with $N=T=40$, where $f_1$ is a standardized AR(1) process with coefficient
$0.6$, $f_2$ is a standardized linear trend, and the loadings and the errors
are standard normal. The $15$ units with the largest values of
$\lambda_{2i}+0.25v_i$, where $v_i$ is standard normal, adopt the treatment in
five cohorts of three units at the dates $20,24,27,31$ and $34$. Selection into
treatment therefore follows exposure to the confounding trend, and the scalar
$\kappa$ controls the strength of the confounding; we consider the four values
$1.5,1,0.5$ and $0.25$, to which we refer as strong, moderate, weak and very
weak. The adoption dates are not all equal, so by
Proposition~\ref{prop:active} the correction is active here; the statistic
\eqref{eq:cstat} takes the value $c=0.085$, so that between eight and nine per
cent of the leading rank-one direction of $W$ falls on cells the weights are
permitted to use. We consider two values of the true average treatment effect
on the treated, $\tau=3$ and $\tau=6$, and we hold the effect constant across the
treated cells. The reason for the latter choice is that the honest
difference-in-differences intervals are estimated on the event-study window
$-4,\dots,3$ and target the average of the post-treatment effects at event
times $0,\dots,3$, whereas the remaining methods target the average over all
treated cells; when the effect is constant the two estimands coincide
and the comparison is not confounded by a difference in the target. Each design
is replicated $500$ times. Monte Carlo standard errors for the coverage and
rejection proportions do not exceed $0.023$; for the remaining summaries they
are at most $0.03$ for bias, $0.10$ for width and $1.02$ for the interval
score. In Section S6 we repeat the study with $\tau=0$, and we explain
there why for every method whose interval width does not depend on the effect
the two designs carry the same information.

We compare the intervals of Section 2 with those obtained from three existing
methods: generalized synthetic control \citep{xu2017}, the
difference-in-differences estimator of \citet{callaway2021} for staggered
adoption, and the honest difference-in-differences procedure of
\citet{rambachan2023}, whose robust confidence set is indexed by a sensitivity
parameter $\bar M$ bounding the permitted violation of parallel trends relative
to the violations observed before adoption. The proposed intervals are reported
at two radii. The first is the oracle radius $\|E\|_*$, which is available only
in a simulation and which we include in order to separate the behaviour of the
minimax construction from the behaviour of the bound supplied to it. The second
is the diagnostic radius $\widehat C$ of Section 2. We emphasise that $\widehat C$ is not proposed as a
default choice in practice. Rather, $C$ is a sensitivity
parameter to be varied and benchmarked, while $\widehat C$
serves only as a diagnostic reference radius whose validity
is not guaranteed in the present setting.

To assess the quality of the obtained intervals we calculated the empirical
coverage probability of the true effect, the proportion of replications in
which the interval excludes zero, the mean width and the mean interval score
\citep{gneiting2007strictly}. The interval score rewards narrow intervals and
adds to the width a penalty of $2/\alpha$ times the distance from the interval
to the true value when that value falls outside, so that low values of the
score are preferable. The quality of the point estimates was assessed through the mean
bias, which is unavailable for honest difference-in-differences since that
procedure returns a set rather than an estimate. We report the proportion of
rejections separately from the interval score because the two are not
equivalent: the score is computed with respect to the true effect, whereas the
proportion of rejections concerns the exclusion of zero, and the two do not
rank the methods in the same order.

Figure~\ref{fig:power} visualizes the evaluation of the $95\%$ intervals for a
true effect of $3$, and Table~\ref{tab:power} summarizes both effect sizes by
reporting, for each method, the range of the coverage probabilities and of the
rejection proportions over the four factor strengths together with the mean
width, interval score and bias. 
The corresponding bias-aware intervals without minimax correction ($A=0$) are reported separately in Table 1, allowing the contribution of the bias allowance to be distinguished from that of the optimized weights.
Panel (a) of the figure shows that the minimax
correction lowers the bias of the imputation estimate it corrects in every cell
of both designs, from $0.59$ to $0.53$ at strong factors under the correct rank
and from $4.01$ to $3.41$ under the underfitted one, the improvement being
modest but uniform. The comparison which the correction is designed to make is
the one against generalized synthetic control, since that is the imputation to
which the weights are applied. The estimator of \citet{callaway2021} attains a
lower bias than either under the underfitted rank, $1.99$ against $3.41$ at
strong factors, which is unsurprising because it fits no factor model and
therefore has none to misspecify; its bias is identical in the two columns for
the same reason. That advantage does not carry over to the intervals, which
cover the true effect in $0.10$ to $0.87$ of the replications. Panel (b) shows the behaviour that motivates the paper. The
proposed intervals attain a coverage probability of one throughout, whereas the
coverage of generalized synthetic control declines steadily as the confounding
factor weakens, from $0.70$ to $0.03$ when the rank is correctly specified and
to at most $0.01$ when it is underfitted; the intervals of
\citet{callaway2021} cover between $0.10$ and $0.87$. Panel (c) shows that these two methods nevertheless reject the null of no effect in at least 99\% of replications in every cell. However, their substantial undercoverage indicates that these high rejection rates cannot be interpreted as evidence of reliable detection. We therefore exclude them from the power comparisons that follow.

It is worth separating what the weights contribute from what the bias allowance
contributes, since the two enter the interval in different ways and the
coverage just reported could in principle be due to either.
Table~\ref{tab:power} therefore also reports, at each of the two radii, the
interval obtained by setting $\widehat A=0$ while retaining the allowance
$C\|W\|_{\mathrm{op}}$, which by Proposition~\ref{prop:active} is the interval
the method returns when adoption is not staggered. Because the radius is held
fixed across the pair, the difference between that row and the proposed row
above it is the contribution of the weights and of nothing else. Its coverage
probability is equal to one in every cell of both designs, so that the
protection against undercoverage is supplied by the allowance rather than by
the weights. What the weights supply is width and bias. The corrected
intervals are narrower by about $13$ per cent at every factor strength and at
both radii, a mean width of $6.4$ against $7.3$ at the oracle radius, and the
point estimate moves against the imputation it corrects by between $10$ and
$15$ per cent of that estimator's bias, which is the difference between
$0.65$ and $0.74$ in the final column of Table~\ref{tab:power}. The local measure $c=0.085$ and the realized gain are
both modest, which the design leads one to expect, though we do not claim a
quantitative relation between the two on the evidence of a single design.

At $\tau=6$ the rejection proportions are close to one among the methods whose
empirical coverage is at or above the nominal level, so the comparison is more
informative at $\tau=3$. There the proposed intervals at
the diagnostic radius detect the effect in $8$ to $29\%$ of the replications
under the correct rank, against $48$ to $65\%$ for honest
difference-in-differences at $\bar M=0.5$. The shortfall is attributable to the
radius rather than to the construction, since the same intervals evaluated at
the oracle radius detect the effect in $85$ to $100\%$ of the replications; a
radius which exceeds the realized error by a factor of $1.25$ to $1.34$, as
$\widehat C$ does here, is what an effect of that size does not survive.

The comparison with honest difference-in-differences at a single value of
$\bar M$ should, however, be interpreted with care, since the power of that
procedure is largely determined by the value assumed. Repeating the design over
$\bar M\in\{0,0.5,1\}$, in a separate run of $200$ replications, we obtained
rejection proportions of $1.000$, of $0.46$ to $0.66$, and of $0.010$ to
$0.020$ respectively, with mean widths of $1.87$, $6.45$ and $11.40$; the
$\bar M=0.5$ figures agree with those of Table~\ref{tab:power} to within Monte
Carlo error. At $\bar M=0$ the unit rejection proportion is accompanied
by coverage of only $0.79$ to $0.93$, so it is of the same character as that of
the two methods excluded above; at $\bar M=1$ the intervals are wider than
those of the proposed method at the diagnostic radius and detect the effect an
order of magnitude less often. The ordering of the two methods in Panel (c) of
Figure~\ref{fig:power} therefore reverses within the range of values of
$\bar M$ that an analyst might defensibly adopt, which is an argument for
reporting a sensitivity parameter against a benchmark, as we do for $C$ in
Section 4, rather than a finding about either procedure.

When the rank is underfitted the ordering of the two radii is reversed, for a
reason which is a limitation of the diagnostic bound rather than an improvement
in the estimator. Since $\widehat C=3R\,s_1(\widehat U)$ is evaluated at the
fitted rank, underfitting halves the multiplier, and the omitted factor does
not inflate $s_1(\widehat U)$ sufficiently to compensate: the diagnostic radius
falls below $\|E\|_*$ in $87$ to $100\%$ of the replications and averages
$0.77$ to $0.88$ times it, against $1.25$ to $1.34$ under the correct rank. The
shorter intervals which result detect an effect of $3$ in $81$ to $100\%$ of
the replications, but this reflects the failure of the bound and not the
performance of the construction. Evaluated at the oracle radius, where the
bound holds by construction, the same design yields a rejection proportion of
$0.36$ at strong factors and of $0.98$ to $1.00$ at the remaining three.
Coverage remains equal to one throughout, because the alignment of $E$ which
makes the worst-case bias attainable does not arise in this design; this is a
property of the simulation rather than a guarantee, and it is the clearest
evidence we have that $C$ must be treated as an assumption and not as a
computed quantity.

Table~\ref{tab:power} indicates that the interval scores rank the methods in a
manner consistent with the preceding discussion without favouring the proposed
intervals on average. Under the correct rank honest difference-in-differences
attains a score of $6.44$, against $6.37$ and $8.06$ for the proposed intervals
at the oracle and the diagnostic radius respectively, so the three are close.
Under the underfitted rank the scores separate the two families sharply,
generalized synthetic control attaining $62.1$ against $8.6$ and $7.2$. In summary, both the proposed intervals and the
uncorrected imputation intervals with the same bias
allowance retain full empirical coverage across the
designs considered. The protection against undercoverage
therefore comes primarily from accounting for approximation
bias, while the minimax correction improves efficiency,
reducing interval length by approximately $13\%$ relative
to the corresponding uncorrected bias-aware intervals.
The resulting intervals remain substantially wider than
conventional GSC intervals, reflecting the explicit bias
allowance. Under the correctly specified factor rank,
the diagnostic reference radius exceeds the realized
approximation error and reduces the probability of
detecting a nonzero effect relative to the oracle radius.
Under rank underfitting, however, the diagnostic radius
may fall below the realized error and cannot be treated
as a valid bound. The comparators based on parallel
trends are themselves sensitive to the way in which
post-treatment periods are aggregated, which the
Supplementary Material documents.

\begin{table}[H]
\centering\footnotesize
\renewcommand{\arraystretch}{0.86}
\caption{Nonzero-effect simulations, $N=T=40$, $500$ replications per cell. Cov and Pow give the range across the four factor strengths of coverage of the true effect and of the share of replications whose interval excludes zero; Len, IS and Bias are means over those cells, IS being the interval score scored against the true effect. The $A=0$ rows retain the bias allowance at the radius named but set the minimax weights to zero, so the comparison with the proposed rows holds $C$ fixed and isolates what the weights contribute. Honest DiD returns a set, not a point estimate. In $2$ of the $500$ replications the honest difference-in-differences set was empty, the moment inequalities being mutually incompatible at $\bar M=0.5$; such a set covers nothing and excludes zero, and is omitted from that row's width and score only. Monte Carlo standard errors for coverage and rejection proportions are at most $0.023$.}
\label{tab:power}
\begin{tabular}{lllccccc}
\toprule
ATT & Rank & Method & Cov & Pow & Len & IS & Bias \\
\midrule
$\tau=3$ & Correct rank & Proposed (oracle) & 1.00--1.00 & 0.85--1.00 & 6.37 & 6.37 & 0.65 \\
 &  & Proposed (diag. radius) & 1.00--1.00 & 0.08--0.29 & 8.06 & 8.06 & 0.65 \\
 &  & GSC + allowance ($A=0$, oracle) & 1.00--1.00 & 0.37--0.93 & 7.34 & 7.34 & 0.74 \\
 &  & GSC + allowance ($A=0$, diag.) & 1.00--1.00 & 0.00--0.02 & 9.34 & 9.34 & 0.74 \\
 &  & GSC & 0.03--0.70 & 1.00--1.00 & 1.23 & 12.22 & 0.74 \\
 &  & C\&S DiD & 0.10--0.87 & 0.99--1.00 & 2.14 & 15.74 & 1.07 \\
 &  & Honest DiD & 1.00--1.00 & 0.48--0.65 & 6.44 & 6.44 & -- \\
\cmidrule(l){2-8}
 & Rank underfitted & Proposed (oracle) & 1.00--1.00 & 0.36--1.00 & 8.61 & 8.61 & 1.93 \\
 &  & Proposed (diag. radius) & 1.00--1.00 & 0.81--1.00 & 7.22 & 7.22 & 1.93 \\
 &  & GSC + allowance ($A=0$, oracle) & 1.00--1.00 & 0.06--1.00 & 9.95 & 9.95 & 2.27 \\
 &  & GSC + allowance ($A=0$, diag.) & 1.00--1.00 & 0.71--1.00 & 8.31 & 8.32 & 2.27 \\
 &  & GSC & 0.00--0.01 & 1.00--1.00 & 1.52 & 62.11 & 2.27 \\
 &  & C\&S DiD & 0.10--0.87 & 0.99--1.00 & 2.14 & 15.74 & 1.07 \\
 &  & Honest DiD & 1.00--1.00 & 0.48--0.65 & 6.44 & 6.44 & -- \\
\midrule
$\tau=6$ & Correct rank & Proposed (oracle) & 1.00--1.00 & 1.00--1.00 & 6.36 & 6.36 & 0.65 \\
 &  & Proposed (diag. radius) & 1.00--1.00 & 1.00--1.00 & 8.07 & 8.07 & 0.65 \\
 &  & GSC + allowance ($A=0$, oracle) & 1.00--1.00 & 1.00--1.00 & 7.33 & 7.33 & 0.74 \\
 &  & GSC + allowance ($A=0$, diag.) & 1.00--1.00 & 1.00--1.00 & 9.34 & 9.34 & 0.74 \\
 &  & GSC & 0.04--0.74 & 1.00--1.00 & 1.22 & 12.27 & 0.74 \\
 &  & C\&S DiD & 0.10--0.87 & 1.00--1.00 & 2.15 & 15.98 & 1.08 \\
 &  & Honest DiD & 0.94--0.95 & 0.98--0.99 & 5.52 & 6.87 & -- \\
\cmidrule(l){2-8}
 & Rank underfitted & Proposed (oracle) & 1.00--1.00 & 1.00--1.00 & 8.59 & 8.59 & 1.90 \\
 &  & Proposed (diag. radius) & 1.00--1.00 & 0.98--1.00 & 7.28 & 7.28 & 1.90 \\
 &  & GSC + allowance ($A=0$, oracle) & 1.00--1.00 & 0.98--1.00 & 9.92 & 9.92 & 2.25 \\
 &  & GSC + allowance ($A=0$, diag.) & 1.00--1.00 & 0.97--1.00 & 8.38 & 8.38 & 2.25 \\
 &  & GSC & 0.00--0.04 & 1.00--1.00 & 1.55 & 61.16 & 2.25 \\
 &  & C\&S DiD & 0.10--0.87 & 1.00--1.00 & 2.15 & 15.98 & 1.08 \\
 &  & Honest DiD & 0.94--0.95 & 0.98--0.99 & 5.52 & 6.87 & -- \\
\bottomrule
\end{tabular}
\end{table}

\begin{figure}[H]
\centering
\includegraphics[width=0.92\linewidth]{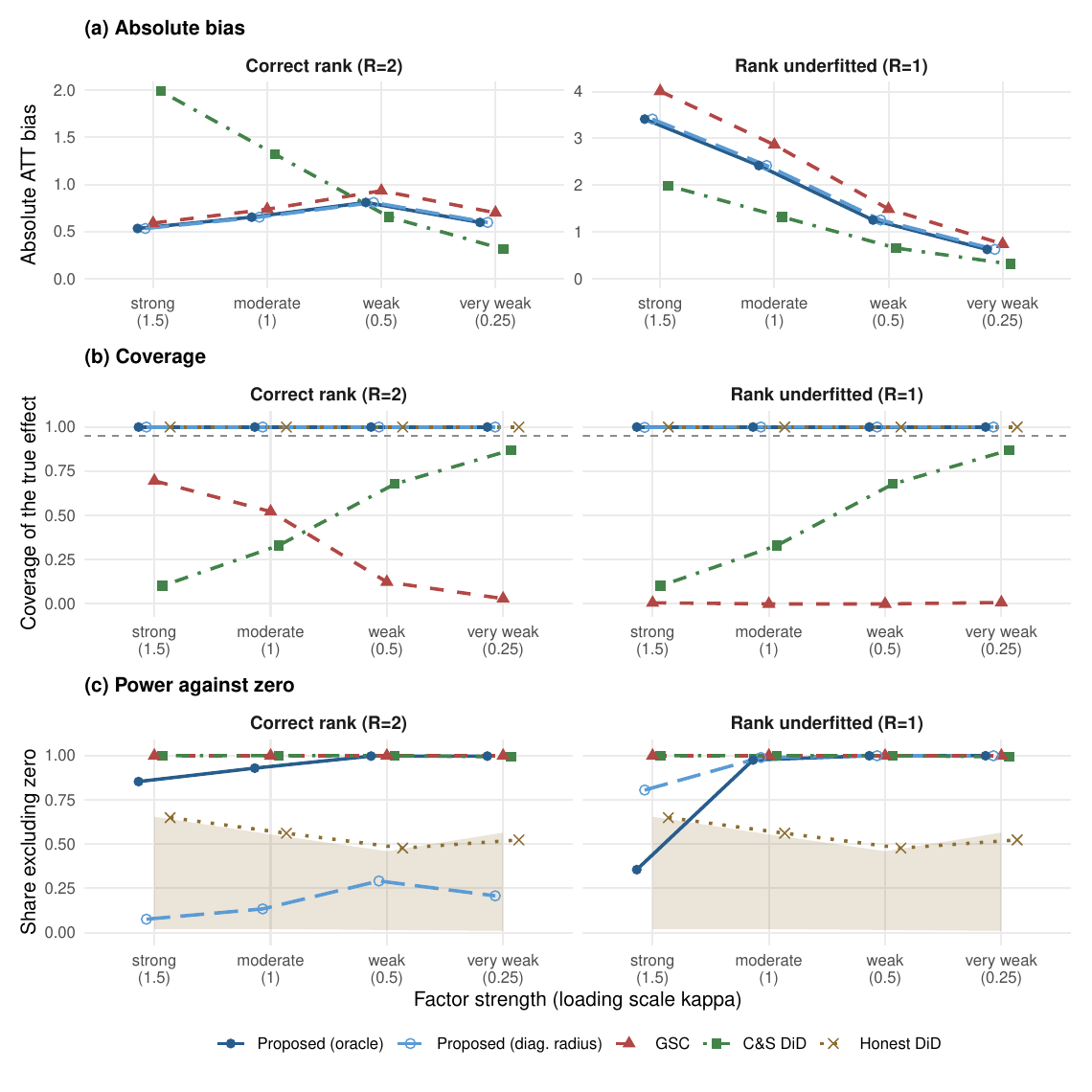}
\caption{Absolute bias (first row), empirical coverage probability of the true
effect (second row) and proportion of replications whose interval excludes zero
(third row), calculated for each of the methods under comparison across the
four factor strengths, for a true ATT of $3$. The left column reports the
designs in which the true rank is fitted and the right column those in which
one factor is fitted against a rank-two truth. The dashed line in the second
row is the nominal level of $95\%$. The two proposed arms are visually
indistinguishable in the first row; the weights solve \eqref{eq:minimax} at a
given $C/\widehat\sigma$, so the point estimate depends on the radius as well,
but over these designs the two differ by at most $0.005$. Honest
difference-in-differences does not appear there since it returns a set
rather than an estimate. In the second and third rows the shaded band spans the
values obtained for that procedure as $\bar M$ ranges over $0.5$ to $1$, the
settings at which its intervals cover, with the line drawn at $\bar M=0.5$; its
position in the third row is therefore an assumption rather than a property of
the method. Table~\ref{tab:power} reports interval scores and the corresponding results for $\tau=6$.}
\label{fig:power}
\end{figure}

\section{Real Data Analysis}

We revisit the Election Day Registration (EDR) application of \citet{xu2017}:
turnout for $47$ states over $24$ presidential elections, nine adopting in four
cohorts, with mail-in and motor-voter controls, two-way effects and $R=2$. Because the actual effect is unknown, we first measure bias on $50$ confounded donor-state placebos with known zero effect, assigning earlier placebo adoption
to the larger pre-1976 turnout trend. Table~\ref{tab:edr} reports the outcome.
The proposed estimator improves on GSC in bias, MAE and RMSE and covers at
$1.00$; GSC covers at $0.96$ here, so in this design the gain over it lies in
the point estimate rather than in coverage, while Callaway--Sant'Anna is badly
biased and covers $0.28$. The price is intervals three to five times wider, and
the interval score prices it against us: generalized synthetic control attains
$10.8$ against $37.8$ for the proposed intervals, since a method which already
covers at $0.96$ gains nothing from the additional width, whereas
Callaway--Sant'Anna attains $56.8$. This is the design in which the
counterfactual error is mild enough that generalized synthetic control covers
zero in $48$ of the $50$ placebo assignments, and the simulations of Section 3
indicate where that ceases to be so.

\begin{table}[H]
\centering\footnotesize
\renewcommand{\arraystretch}{0.9}
\caption{Fifty confounded donor-state placebos with a known zero effect, so Bias is the mean estimate. Cov is coverage of the nominal 95\% interval, Len its mean length and IS the mean interval score of \citet{gneiting2007strictly}, scored against the true effect of zero. The Proposed row is evaluated at the diagnostic reference radius $\widehat C$ computed on each placebo assignment. Monte Carlo standard errors are at most 0.33 for bias, 0.24 for RMSE, 0.06 for coverage and 7.57 for the score.}
\label{tab:edr}
\begin{tabular}{lcccccc}
\toprule
Method & Bias & MAE & RMSE & Cov & Len & IS \\
\midrule
Proposed & -0.79 & 1.94 & 2.37 & 1.00 & 37.80 & 37.80 \\
GSC & -0.88 & 2.02 & 2.45 & 0.96 & 10.51 & 10.83 \\
C\&S DiD & 4.62 & 4.62 & 4.93 & 0.28 & 7.36 & 56.81 \\
\bottomrule
\end{tabular}
\end{table}

For the actual policy GSC estimates an ATT of $4.90$ points with interval
$[0.29,9.50]$. The correction leaves the estimate at $4.99$ and, as
Figure~\ref{fig:edr} shows, it barely moves across the sweep, from $4.895$ to
$4.989$ as $C$ runs from $0$ to $400$: the weights saturate, so the radius
governs the interval rather than the estimate.
The half-width increases with $C$ while the point
estimate changes little, yielding a largest radius
at which the interval still excludes zero, which
we call the breakdown radius. For the uncorrected
imputation estimator with the same bias allowance
($A=0$), this radius is $32.7$, whereas the minimax
correction increases it to $37.9$, an improvement
of approximately $16\%$. Thus, although the correction
has little effect on the point estimate, it increases
the magnitude of assumed counterfactual error under
which the estimated treatment effect remains
statistically significant. Whether this improvement
is substantively meaningful depends on the benchmark
used to assess the sensitivity radius. The placebo effect is zero, so a placebo interval which excludes zero
is a false positive, and because the half-width increases in $C$ while the
estimate does not, each placebo has a largest radius at which it still does so.
We solve for that radius on each of the $50$ placebo assignments of
Table~\ref{tab:edr}, which is the same root the breakdown solves, evaluated on
a design where the answer is known. For the minimax-corrected procedure, the resulting
placebo breakdown radii have median $7.7$ and
$95$th percentile $31.2$; a fifth of the placebos
already cover zero at $C=0$, and $49$ of the $50$,
or $98\%$, have breakdown radii below $37.9$. These
radii provide a descriptive, design-specific
calibration of the sensitivity analysis, rather
than estimates of the nuclear norm of counterfactual
approximation error or a formal coverage guarantee.
The same calculation applied to the uncorrected
estimator gives placebo radii with median $6.8$ and
$95$th percentile $28.4$, of which $48$ of the $50$
fall below its own breakdown radius of $32.7$. The
like-for-like comparison is therefore between
margins of $1.21$ and $1.16$ over the respective
$95$th percentiles, a smaller advantage than the
rise from $32.7$ to $37.9$ on its own suggests,
because correcting the estimator shifts the placebo
distribution as well as the breakdown radius.
At the radius at which the Election Day Registration conclusion fails, that is,
one placebo assignment in fifty would still exclude zero. The shaded band of
Figure~\ref{fig:edr} spans the typical and the $95$th-percentile placebo
magnitudes, $21.3$ to $38.6$, obtained instead by dividing those magnitudes by
$\|\widehat H\|_{\mathrm{op}}$ evaluated at the breakdown radius of the
application, the same operator norm the allowance there uses. The two routes
agree in direction and we report both; neither recovers $\|E\|_*$. On the scale
of the bias itself the same comparison reads as a typical counterfactual bias
of $2.33$ points against the $4.15$ which the interval absorbs at the
breakdown. At the diagnostic reference radius
$\widehat C=266.0$, which is not a theoretically
validated bound for the present design, the
confidence interval includes zero.
Honest
DiD is less encouraging, its breakdown $\bar M$ being $0$, since even under
exact parallel trends the robust set $[-0.98,4.63]$ contains zero.

\begin{figure}[H]
\centering
\includegraphics[width=0.9\linewidth]{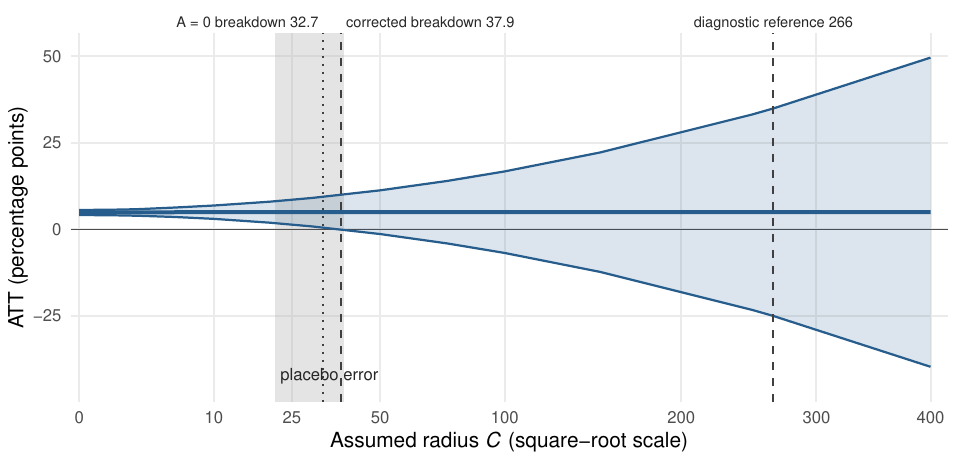}
\caption{EDR sensitivity analysis against the assumed
radius $C$, on a square-root horizontal scale. The solid
blue line is the point estimate and the shaded blue region
the $95\%$ interval, both as functions of $C$; the
horizontal line at zero marks the null. The
shaded band spans the typical and
$95$th-percentile placebo errors, $21.3$ to $38.6$
after scaling by $\|\widehat H\|_{\mathrm{op}}$; it is a
descriptive benchmark for this placebo design and not a
bound on $\|E\|_*$. The uncorrected imputation estimator
with the same bias allowance ($A=0$) has breakdown radius
$32.7$, marked by the dotted line, while the
minimax-corrected estimator has breakdown radius $37.9$,
marked by the dashed line immediately to its right.
The third, rightmost marker is the diagnostic reference
radius $\widehat C=266.0$, at which the interval contains
zero.}
\label{fig:edr}
\end{figure}

\section{Conclusion}

We extend minimax bias-aware inference from panel estimation to causal
inference with systematically missing untreated potential outcomes. The
resulting estimator retains the ATT targeted by generalized synthetic control
while correcting the imputed counterfactual with weighted residuals from
observed untreated cells. The weights explicitly trade off worst-case
counterfactual bias against sampling variance. This also reveals a feature of
the treatment design. Proposition~\ref{prop:active} shows that the correction
is active precisely when the adoption dates are not all equal, and that it
otherwise reduces to the uncorrected imputation estimator; the share of the
leading direction of $W$ lying on observed cells grades the available
improvement and is computable from the assignment pattern before any outcome
is used.

In simulations, the resulting intervals protect against the severe
undercoverage observed for GSC, particularly when the factor rank is
underfitted, at the cost of substantially greater width. The empirical
application illustrates a complementary use of the method as a sensitivity
analysis. Rather than requiring counterfactual approximation error to be
negligible implicitly, the procedure expresses it through a single radius
$C$, which can be varied, benchmarked against placebo errors, and summarized
by the largest value at which a substantive conclusion survives. This parallels
the role of sensitivity parameters in honest difference-in-differences, but
applies the idea to counterfactual imputation error in factor-model panels.

The remaining limitation is that $C$ must be supplied rather than derived for
the incomplete-panel design considered here, and Section 3 shows the cost of
substituting a diagnostic for a derivation: the diagnostic radius exceeds the realized
error under the correctly specified rank but
frequently falls below it under rank underfitting, in the very case the method is
meant to handle. Extending the error bound of
\citet{armstrong2022} to a panel whose treated cells are systematically
unobserved would turn the high-level coverage result into a primitive guarantee
for this setting. Such a result could also yield a less conservative,
design-specific radius than the diagnostic reference
used here. Section 3 illustrates the potential gains:
under the correctly specified rank, at a true effect
of $3$, the proposed intervals detect the effect in
at most $29\%$ of replications with the diagnostic
radius, compared with at least $85\%$ at the oracle
radius. This difference highlights the importance
of calibrating the approximation-error bound, rather
than a limitation of the minimax correction itself. Until then, the
proposed procedure is best viewed as bias-aware inference conditional on an
explicit and interpretable bound on counterfactual error.

\setlength{\bibsep}{1pt plus 0.3ex}

\clearpage
\setcounter{figure}{0}
\setcounter{table}{0}
\renewcommand{\thefigure}{S\arabic{figure}}
\renewcommand{\thetable}{S\arabic{table}}

\begin{center}
{\large\bfseries Supplementary material}
\end{center}
\medskip

\section*{S1. Setup, notation, and estimand}

Matrices have $T$ rows (periods) and $N$ columns (units), with Frobenius inner
product $\langle B,G\rangle=\operatorname{tr}(B'G)$. The nuclear norm
$\|B\|_*=\sum_j s_j(B)$ sums the singular values, $\|B\|_{\mathrm{op}}=s_1(B)$
is the largest, and $\|B\|_F=\{\langle B,B\rangle\}^{1/2}$.

Let $D_{ti}$ equal one when unit $i$ is treated at date $t$, and set
$\mathcal T=\{(t,i):D_{ti}=1\}$, $\mathcal O=\{(t,i):D_{ti}=0\}$. Under
staggered adoption $\mathcal O$ contains every observation for never-treated
units and every pre-adoption observation for eventually treated units, so the
first stage uses both. Outcomes obey
\[
Y_{ti}=\Gamma_{ti}+\tau_{ti}D_{ti}+\varepsilon_{ti},\qquad
\Gamma_{ti}=x_{ti}'\beta+\alpha_i+\delta_t+L_{ti},\qquad
\operatorname{rank}(L)\leq R ,
\]
where $\Gamma_{ti}=E\{Y_{ti}(0)\mid x_{ti},\alpha_i,\delta_t,L_{ti}\}$ and
$\tau_{ti}=Y_{ti}(1)-Y_{ti}(0)$. Four assumptions are maintained throughout.
Consistency: the observed outcome is $Y_{ti}(D_{ti})$. No interference: unit
$i$'s outcome does not depend on other units' treatments, so $\tau_{ti}$ is
well defined. No anticipation: $Y_{ti}=Y_{ti}(0)$ at every $(t,i)\in\mathcal
O$, which is what makes the pre-adoption periods of eventually treated units
usable in the first stage. And exogeneity of the idiosyncratic error,
$E(\varepsilon\mid D,x,\alpha,\delta,L)=0$. The last is weaker than parallel
trends and is the substantive one: adoption may be driven by the latent
untreated-outcome structure, in particular by the factor loadings, which is the
selection this design is built for, but not by $\varepsilon$.

Two further conditions play distinct roles and should not be conflated. The
minimax criterion of S3 treats the errors as uncorrelated with common variance
$\sigma^2$, which is what reduces
$\operatorname{Var}\langle H,\varepsilon\rangle$ to $\sigma^2\|H\|_F^2$ and
makes the program a second-order cone problem. Inference does not rest on that
form directly: Proposition~2 assumes only that
$\langle\widehat H,\varepsilon\rangle/\widehat{\mathrm{se}}$ is asymptotically
standard normal. If the errors are serially correlated the criterion is merely
solving for the wrong variance and loses efficiency, while
$\widehat\sigma\|\widehat H\|_F$ is the wrong scale and would have to be
replaced by a dependence-robust standard error; the proposition is stated in
terms of $\widehat{\mathrm{se}}$ and goes through with that substitute
unchanged. In the designs reported here the substitution would barely register,
since $b/\widehat{\mathrm{se}}$ of $42$--$45$ leaves the interval almost
entirely bias allowance.

Effects
may vary across units and periods; with $W=D/|\mathcal T|$ the estimand is
\[
\tau_W=\langle W,\tau\rangle=\frac{1}{|\mathcal T|}
\sum_{(t,i)\in\mathcal T}\tau_{ti},
\]
the ATT averaged over treated post-adoption observations. This is the substantive
difference from the regression-coefficient problem for which the minimax
bias-aware apparatus was originally built: the target is defined by imputation
over untreated potential outcomes that are systematically missing, it
aggregates effects that may differ across treated cells, and corrections are
admissible only on observed untreated cells. The object whose error must be
controlled is a counterfactual matrix in both settings.

\section*{S2. First-stage interactive-fixed-effects estimation}

For a chosen rank $R$,
\[
(\widehat\beta,\widehat\alpha,\widehat\delta,\widehat L)
=\arg\min_{\beta,\alpha,\delta,L:\,\operatorname{rank}(L)\leq R}
\sum_{(t,i)\in\mathcal O}
\left(Y_{ti}-x_{ti}'\beta-\alpha_i-\delta_t-L_{ti}\right)^2 .
\]
Normalizations such as $T^{-1}F'F=I_R$ with diagonal $\Lambda'\Lambda$ identify
factors and loadings up to sign without changing the completed $\widehat L$.
The fitted mean
$\widehat\Gamma_{ti}=x_{ti}'\widehat\beta+\widehat\alpha_i+\widehat\delta_t
+\widehat L_{ti}$ is evaluated at every observation, including $\mathcal T$; no
outcome with $D_{ti}=1$ enters. The implementation uses the alternating
least-squares IFE routine in \texttt{gsynth}. Running the conventional
comparator through the same routine ensures differences do not come from
different factor estimates. The same fit supplies $\widehat\sigma^2$, the mean
squared residual on $\mathcal O$ defined in the main text.

\section*{S3. Minimax weights and computation}

Under homoskedastic errors the weights solve
\[
\widehat A=\arg\min_{A:\,A_{\mathcal T}=0}
\left\{C^2\|W-A\|_{\mathrm{op}}^2+\widehat\sigma^2\|A\|_F^2\right\},
\]
the second term being the part of
$\operatorname{Var}\{\langle W-A,\varepsilon\rangle\}
=\sigma^2(\|W\|_F^2+\|A\|_F^2)$ that depends on $A$, with $\sigma^2$
replaced by the residual-based $\widehat\sigma^2$; the four steps below derive
the benchmark at a known $\sigma$ and the paragraph following them states what
that substitution does and does not preserve. The squared operator norm
is convex and the support restriction is linear, so this is a second-order cone
program, solved with \texttt{CVXR} and SCS at tolerance $10^{-5}$. The
objective is rescaled by $\max(C^2,\widehat\sigma^2)$ before solving; without
this SCS returns a spurious ``unbounded'' status once $C/\widehat\sigma$ is
large, which matters because the breakdown calculation sweeps $C$ over several
orders of magnitude. Only the leading singular direction is penalized, and it
can be flattened cheaply, so the correction is active and lowers both bias and
RMSE.

Optional constraints $A\mathbf 1_N=W\mathbf 1_N$ and
$A'\mathbf 1_T=W'\mathbf 1_T$ make the correction balance additive time and
unit components exactly. They force $H\mathbf 1=0$ and $H'\mathbf 1=0$, which
annihilates the additive part of $E$ exactly, and might therefore be expected
to help. They are not imposed in the reported results because they do not:
restricting $A$ that far raises $\|H\|_F$ by more than the tighter bias bound
saves, and over $40$ draws of the design of Section 3 the interval half-width
rises from $1.94$ to $2.93$, $51$ per cent wider.

The main text states that the program minimizes the plug-in version of a
worst-case mean squared error computed at fixed weights. That benchmark
follows in four steps, in each of which $A$ is a deterministic admissible
matrix and $E$ is held fixed. Write $H=W-A$, $b(A)=C\|H\|_{\mathrm{op}}$,
$\mathrm{se}(A)=\sigma\|H\|_F$ at the true noise scale and
$\widehat{\mathrm{se}}(A)=\widehat\sigma\|H\|_F$ for its plug-in counterpart.

\emph{Step 1.} In the decomposition
$\widehat\tau(A)-\tau_W=\langle H,E\rangle+\langle H,\varepsilon\rangle$ the
first term is constant given the weights and the second has mean zero, so the
cross term vanishes on squaring and taking expectations:
\[
\operatorname{MSE}(A,E)=\langle H,E\rangle^2+\sigma^2\|H\|_F^2 .
\]

\emph{Step 2.} $E$ is unknown and only its membership of $\mathcal E(C)$ is
assumed, so the relevant quantity is the largest this can be over that class.
The variance carries no $E$, so the maximum moves the first term alone, and
duality evaluates it exactly:
\[
\sup_{E\in\mathcal E(C)}\operatorname{MSE}(A,E)
=C^2\|H\|_{\mathrm{op}}^2+\sigma^2\|H\|_F^2
=b(A)^2+\mathrm{se}(A)^2 .
\]
This is where $b$ comes from: it is not a separate construction but what the
maximization returns.

\emph{Step 3.} $W$ and $A$ have disjoint supports, so
$\|H\|_F^2=\|W\|_F^2+\|A\|_F^2$ and the variance splits into a part that
depends on $A$ and one that does not:
\[
\sup_{E\in\mathcal E(C)}\operatorname{MSE}(A,E)
=\underbrace{C^2\|W-A\|_{\mathrm{op}}^2+\sigma^2\|A\|_F^2}
_{\text{the part that depends on }A}
\;+\;\underbrace{\sigma^2\|W\|_F^2}_{=\ \sigma^2/|\mathcal T|} .
\]

\emph{Step 4.} The final term is the noise any estimator of $\tau_W$ carries
through the treated observations, fixed by the design and free of $A$. It
shifts the criterion without moving its minimizer, and is dropped. At a known
$\sigma$ the remainder is therefore the fixed-weight worst-case mean squared
error less a constant, and the two are minimized at the same point.

The program we solve replaces $\sigma^2$ by $\widehat\sigma^2$, and it is in
that sense a plug-in criterion rather than an exact mean squared error for the
estimator as implemented. The distinction matters because $\widehat\sigma$,
the fitted counterfactual, and at the diagnostic reference radius the radius
itself are all computed from the data, so $\widehat A$ is random, whereas the
four steps above hold $A$ fixed. We therefore do not claim that
$\sup_E\operatorname{MSE}$ evaluated at $\widehat A$ equals the unconditional
mean squared error of $\widehat\tau(\widehat A)$. What does survive the
substitution is the decomposition of Step 1 and the duality bound of Step 2,
both of which hold when evaluated at the realized weights; the coverage
statement for the implemented procedure is Proposition 2 of the main text,
whose conditions are maintained rather than verified here.

Two further properties of the objective are worth recording. It is homogeneous
of degree two in $(C,\widehat\sigma)$, so $\widehat A$ depends on the two only
through $C/\widehat\sigma$; solved at $(70,1)$ and at $(140,2)$ it returns
bit-identical weights. And the solution saturates in $C$. In the design of
Section 3, where $\|W\|_{\mathrm{op}}=0.0638$, the norm $\|\widehat A\|_F$
rises from zero at $C=0$ to $0.055$ at $C=5$ and $0.064$ at $C=400$, while
$\|\widehat H\|_{\mathrm{op}}$ falls to $0.0542$ by $C=5$ and is constant
thereafter: with $A$ confined to $\mathcal O$ the operator norm cannot be
reduced further. The bias bound $b=C\|\widehat H\|_{\mathrm{op}}$ is therefore
very nearly linear in $C$ beyond a small threshold, which is why the breakdown
radius is well defined and why the point estimate in Section 4 moves so little
across the sweep.

The trade-off the program resolves can then be read off directly. The two
ingredients of Step 2 move against each other: a larger correction lowers $b$
by flattening the leading singular direction of $W$, and raises
$\widehat{\mathrm{se}}$ by importing untreated residuals.
Table~\ref{tab:tradeoff} traces both along the ray $A=t\widehat A$ in the
design of Section 3, at $C=73.9$ and $\sigma=1$.

\begin{table}[H]
\centering\small
\caption{Worst-case bias, noise and their squared sum along $A=t\widehat A$,
five-cohort design, $C=73.9$, $\sigma=1$. The correction solves the program at
$t=1$.}
\label{tab:tradeoff}
\begin{tabular}{lccccccc}
\toprule
$t$ & 0 & 0.25 & 0.50 & 0.75 & 1 & 1.25 & 1.50 \\
\midrule
$\|A\|_F$ & 0.000 & 0.016 & 0.032 & 0.047 & 0.063 & 0.079 & 0.095 \\
$b$ & 4.715 & 4.467 & 4.257 & 4.093 & 4.008 & 4.454 & 5.182 \\
$\widehat{\mathrm{se}}$ & 0.0695 & 0.0713 & 0.0763 & 0.0841 & 0.0939 & 0.1052
& 0.1175 \\
$b^2+\widehat{\mathrm{se}}^2$ & 22.24 & 19.95 & 18.13 & 16.76 & 16.07 & 19.85
& 26.87 \\
objective & 22.23 & 19.95 & 18.12 & 16.75 & 16.06 & 19.84 & 26.86 \\
\bottomrule
\end{tabular}
\end{table}

Too little correction leaves avoidable bias, too much imports noise that
outweighs what it removes, and $t=1$ balances the two. It is worth being
explicit about the scale on which that balance is struck. At $t=1$ the squared
bias bound is $16.06$ against $\widehat{\mathrm{se}}^2=0.0088$, so the variance
is $0.055$ per cent of the objective: in this regime the program is
operator-norm minimization for practical purposes, and the variance term
resolves ties rather than driving the solution. That is a property of the
design and the radius, not of the criterion. The variance would bind at a
radius small enough for $C\|H\|_{\mathrm{op}}$ and
$\widehat\sigma\|H\|_F$ to be comparable, which here means $C$ of order
$\widehat\sigma\|H\|_F/\|H\|_{\mathrm{op}}\approx1.7$ rather than the $73.9$
used above, or in a panel where $|\mathcal T|$ is small enough that
$\|W\|_F$ is large. The last two rows differ
by $\widehat\sigma^2\|W\|_F^2=1/207=0.0048$ at every $t$, which is Step 4
confirmed numerically. The name minimax is exact here --- the maximum is over
$E\in\mathcal E(C)$ and produces $b^2$, the minimum is over $A$ and produces
the weights --- and the construction has no tuning parameter beyond $C$, which
is reported as a sensitivity parameter in any case.

Minimizing interval length rather than mean squared error is a natural
alternative, and in general the two criteria need not select the same weights.
Here they very nearly do. Along the same ray the reported half-length
$b+z_{1-\alpha/2}\widehat{\mathrm{se}}$ falls from $4.851$ at $t=0$ to $4.192$
at $t=1$, then rises to $4.269$ at $t=1.1$ and $5.412$ at $t=1.5$; on a grid of
width $0.005$ it is minimized at $t=0.98$, against $t=1$ for the worst-case
mean squared error. The sharp constant
$\widehat{\mathrm{se}}\,cv_\alpha(b/\widehat{\mathrm{se}})$ gives $4.830$,
$4.162$, $4.238$ and $5.375$ at the same four points and is minimized at the
same $t=0.98$, so the choice of critical value does not move the comparison.
The reason both criteria track each other is that
$b/\widehat{\mathrm{se}}$ averages $42$--$45$ in these designs, so the half-length is
dominated by $b$, which is also what dominates the mean squared error; where
the bias bound and the sampling error were comparable the two could part
company. That is a check along one ray; re-solving the program settles it
globally. The half-length
$C\|H\|_{\mathrm{op}}+z_{1-\alpha/2}\widehat\sigma\|H\|_F$ is convex in $A$,
since both norms are convex and $W-A$ is affine, so minimizing it directly is a
second-order cone program like the one above. Doing so shortens the interval in
the design of Section 3 from $4.192$ to $4.183$, by $0.21$ per cent; on the EDR
application it shortens by $0.09$ per cent at $\widehat C$ and by $0.08$ per
cent at the breakdown bound. Even at $C=1.73$, the radius at which $b$ and
$\widehat{\mathrm{se}}$ are equal and the variance therefore matters most, the
saving is $1.4$ per cent and there the two solutions do differ in the
expected direction, the interval-length weights accepting a larger operator
norm, $0.0602$ against $0.0574$, to buy a smaller Frobenius norm, $0.0709$
against $0.0752$. Armstrong, Weidner and Zeleneev report the same insensitivity
to the criterion in their Remark 2.3. The width of these intervals is set by
the radius, not by what the program optimizes.

\section*{S4. Bias-aware critical value}

With bias bound $b$ and standard error $\widehat{\mathrm{se}}$, the reported
half-length $b+z_{1-\alpha/2}\widehat{\mathrm{se}}$ adds the worst-case bias to
a sampling quantile. A sharper alternative exists, due to Armstrong and
Koles\'{a}r (2018): a bias of known magnitude shifts the sampling distribution
rather than widening it, so a half-length
$\widehat{\mathrm{se}}\,cv_\alpha(m)$ suffices, where $m=b/\widehat{\mathrm{se}}$
and $cv_\alpha(m)$ is the $1-\alpha$ quantile of $|N(m,1)|$. Since $|N(m,1)|^2$
is noncentral chi-squared with one degree of freedom and noncentrality $m^2$,
\[
cv_\alpha(m)=\sqrt{\chi^2_{1,\,m^2;\,1-\alpha}} ,
\]
which is exact, satisfies $cv_\alpha(0)=z_{1-\alpha/2}$ and never exceeds
$m+z_{1-\alpha/2}$. It is not used in what follows; the figures below record
what declining it costs.

The two conditions of Proposition 2 are those of Assumption 1 of Armstrong,
Weidner and Zeleneev (2026): the radius bounds the nuclear-norm error with
probability approaching one, and the studentized noise is asymptotically
standard normal. Their Theorems 2--5 establish both for a scalar coefficient in
a complete panel, under their Assumptions 2 and 3, the second required also
with the covariates interchanged. Neither route is open here: the treated block
is structurally unobserved, and the weights are pinned by the design rather
than by a normalization $\langle A,X\rangle_F=1$. Section S5 therefore reports how the
diagnostic radius compares with the realized error rather than deriving the
first condition.

The second is not reproduced either. Their Theorem 2 presumes the modification
of their Remark 3.3, a bound on the Lindeberg weight
$\max_{t,i}H_{ti}^2/\sum_{t,i}H_{ti}^2$ imposed inside the weight program,
which is what delivers asymptotic normality of $\langle A,U\rangle_F$. That
ratio is not convex, but an elementwise cap $|H_{ti}|\leq\eta$ is, and implies
it. Without a cap the solved weights have Lindeberg weight $0.017$, equivalent
to a sample mean of $59$ observations; at $\eta=0.010$ this falls to $0.013$,
or $79$ observations, while the operator norm rises only from $0.0542$ to
$0.0546$. The reported results impose no cap, so finite-sample coverage is
assessed by simulation instead.

\medskip\noindent\textit{Proof of Proposition 2.} On the event
$\{\|E\|_*\leq C\}$, evaluating the duality identity at $\widehat H$ gives
$|\langle\widehat H,E\rangle|\leq C\|\widehat H\|_{\mathrm{op}}=b$, so the
decomposition $\widehat\tau(\widehat A)-\tau_W=\langle\widehat H,E\rangle
+\langle\widehat H,\varepsilon\rangle$ yields
$|\widehat\tau(\widehat A)-\tau_W|\leq b+|\langle\widehat H,
\varepsilon\rangle|$. An interval of half-width
$b+z_{1-\alpha/2}\widehat{\mathrm{se}}$ can therefore fail to cover only if
$|\langle\widehat H,\varepsilon\rangle|>z_{1-\alpha/2}\widehat{\mathrm{se}}$ or
$\|E\|_*>C$, so
\[
\Pr\{\tau_W\notin\widehat\tau(\widehat A)\pm(b+z_{1-\alpha/2}
\widehat{\mathrm{se}})\}
\leq\Pr\Bigl(\bigl|\langle\widehat H,\varepsilon\rangle\bigr|
/\widehat{\mathrm{se}}>z_{1-\alpha/2}\Bigr)+\Pr(\|E\|_*>C).
\]
The first term tends to $\alpha$ by the assumed studentized limit, since
$z_{1-\alpha/2}$ is a continuity point, and the second vanishes. $\square$

\medskip
Studentizing is what removes the two awkward requirements of a finite-sample
statement. It needs no assumption that $\widehat\sigma$ overstate $\sigma$,
the mean squared residual on $\mathcal O$ carries the counterfactual error
upward and the overfitting of a rank-$R$ fit downward, with neither dominating
by construction, and it needs no claim that the errors stay exactly normal
after conditioning on a variance estimate built from those same errors, which
they do not: conditioning a linear form on a quadratic form in the same
variables perturbs its law.

Neither condition is verified here for the implemented estimator. What the
construction does not require is that $\widehat A$ be chosen independently of
the data, since the weights enter only through $\widehat\sigma$; a sample split
together with an upper confidence bound for $\sigma$ would deliver both
conditions outright, and a primitive same-sample proof would require showing
that the randomness of $\widehat H$ and the leave-in estimation effects are
asymptotically negligible, which we do not attempt. Without either, they are
high-level conditions in the sense of Assumption 1 of Armstrong, Weidner and
Zeleneev (2026), and S6 reports the coverage they deliver: $1.00$ across every
design, with $\widehat\Gamma$, $\widehat\sigma$ and $\widehat A$ all computed
from the same sample as in the application. That they are not binding at these
magnitudes is unsurprising: with $b/\widehat{\mathrm{se}}$ of $42$--$45$ the
interval is almost entirely bias allowance, and doubling $\widehat\sigma$ would
widen it by $2.4\%$. The reported intervals use the conservative half-length
$b+z_{1-\alpha/2}\widehat{\mathrm{se}}$ rather than the sharp
$\widehat{\mathrm{se}}\,cv_\alpha(m)$. Nothing is lost by that choice here. The
saving from the sharp constant is largest ($8$--$11\%$) when $m$ is of order
one and vanishes as $m$ grows, since $cv_\alpha(m)\to m+z_{1-\alpha}$; in the
designs reported here $m$ averages $42$--$45$, so the sharp interval would be
under one percent shorter. On the EDR application the sharp constant would
raise the breakdown bound from $37.9$ to $39.1$ and narrow the interval at
$\widehat C$ by $0.45\%$.
These intervals are almost entirely bias allowance, so the constant that
converts noise into width is not what decides any conclusion.

\section*{S5. The radius, and how it compares with the realized error}

Armstrong, Weidner and Zeleneev (2026), Theorem 2, bound the nuclear-norm error
of the preliminary estimator of their Algorithm 3.1 by
\[
\widehat C = cR\,s_1(\widehat U_{\mathrm{pre}})(1+\epsilon),
\]
with $c=3$ when the target is $\Gamma$ and $c=2$ when it is redefined as
$\Gamma+P_\lambda U$. This procedure sets $E=\Gamma-\widehat\Gamma$, which is
the first case and so requires $c=3$. Using $c=2$, the constant belonging to the
other target, gives a radius that fails to bound $\|E\|_*$ in $87.5$--$100$ per
cent of replications; with $c=3$ it bounds it in $99.5$--$100$ per cent.

The theorem is not a statement about rank-$R$ fits in general. It applies to
$\widehat\Gamma_{\mathrm{pre}}$, which in their Algorithm 3.1 is a second
rank-$R$ fit taken after an initial minimax debiasing pass, and it holds under
their Assumptions 2 and 3, the latter required both as stated and with the
covariates interchanged. The first stage here is a single rank-$R$ least
squares on $\mathcal O$ in a panel whose treated block is never observed
untreated. Table~\ref{tab:radius} therefore reports what the bound does in
these designs, not what their theorem guarantees, and the correctly specified
column is the only one in which the premise that the true rank is at most the
fitted $R$ holds at all.

\begin{table}[H]
\centering\small
\caption{Diagnostic radius against the truth, $200$ replications per design.
``Holds'' is the share of replications with $\widehat C\geq\|E\|_*$.}
\label{tab:radius}
\begin{tabular}{lccccc}
\toprule
Factor strength & $\|E\|_*$ & $\widehat C$ at $2R$ & Holds & $\widehat C$ at $3R$ & Holds \\
\midrule
$\kappa=1.5$  & 56.3 & 49.9 & 0.125 & 74.9 & 1.000 \\
$\kappa=1$    & 56.8 & 49.3 & 0.060 & 73.9 & 0.995 \\
$\kappa=0.5$  & 56.5 & 46.8 & 0.000 & 70.1 & 0.995 \\
$\kappa=0.25$ & 52.3 & 44.3 & 0.010 & 66.4 & 1.000 \\
\bottomrule
\end{tabular}
\end{table}

At $c=3$ the radius averages $1.25$--$1.34$ times $\|E\|_*$, so in this design
it exceeds the realized error rather than falling short of it. That is an
observation about these replications and not a guarantee: as the paragraph
below records, the theorem does not reach this estimator, and the margin
depends on the fitted rank. Section 3 of the main text reports the same
comparison when one factor is fitted against a rank-two truth, where the
radius instead averages $0.77$--$0.88$ times $\|E\|_*$ and falls below it in
$87$--$100$ per cent of replications. The correctly specified column is
therefore the favourable case, not the typical one. Where the radius does
exceed the error, the cost is width: intervals run $7.5$--$8.4$ against
$6.0$--$6.4$ under the infeasible oracle radius, with coverage $1.00$ in both
cases and point estimates unchanged to three digits.

Two cautions. First, their proof assumes a complete panel and a scalar
coefficient; the treated block here is structurally unobserved, so
Table~\ref{tab:radius} is verification, not derivation. Second, $\widehat C$
scales with the noise: at $T=N=40$, $R=2$ and $\sigma=1$ the term
$s_1(\widehat U)$ is bounded below by roughly $\sigma(\sqrt T+\sqrt N)\approx12$
even when $E=0$, so $\widehat C\approx70$ with no counterfactual error at all.
That is the expected behaviour of an upper bound rather than an estimate --- the
same scale appears in their tuning parameter $b^*=2R(\sqrt N+\sqrt T)$ --- but it
means $\widehat C$ does not adapt to how well the model actually fits.
Subtracting the noise level, which would make it adapt, destroys the bound: the
resulting radius is near zero under correct specification while $\|E\|_*$ is
about $55$, because that error lives almost entirely in the treated block where
untreated data cannot see it. Reporting $C$ as a sensitivity parameter alongside
$\widehat C$, and summarising by the breakdown radius, is the response to both
cautions.

For the record, Armstrong, Weidner and Zeleneev raise a heterogeneous-effect
target in their Remark 2.2 and call a computable bound for it nontrivial,
leaving the question open. The radius used here is their nuclear bound
transplanted and checked, not a solution to that problem: what
Table~\ref{tab:radius} establishes is that the transplant survives in these
designs once the constant is right, which is weaker than a derivation and is a
further reason to vary $C$.

\section*{S6. The zero-effect design}

The main text reports simulations with a true ATT of $3$ and of $6$. Repeating
them with $\tau=0$ adds nothing for any method whose interval width is
unrelated to the effect, because the distribution of $\widehat\tau-\tau$ is
then the same at every $\tau$: coverage is unchanged, and power at any $\tau$
follows by shifting the zero-effect estimates. Taking the draws below and
shifting them by $3$ reproduces the power of Table 1 of the main text to within
Monte Carlo error. The design is retained here for three reasons. It reports
bias directly, since with $\tau=0$ the mean estimate is the bias. Its stored
results additionally carry a soft-thresholded radius arm, which neither the
main text nor this supplement reports, retained so that the effect of
subtracting the noise level from $\widehat C$ can be recovered without
refitting. And it is the conventional placebo format, so it makes the
comparison with the published simulations of the methods studied here direct.

Honest difference-in-differences is the one method for which the two designs
genuinely differ. Its set rests on a conditional moment-inequality test whose
critical values depend on which moments bind, so its width is not translation
invariant: between $\tau=3$ and $\tau=6$ its mean length falls by about one
unit at every factor strength, with $t$ statistics of $3.7$ to $5.9$, and its
coverage falls from $1.00$ to $0.94$--$0.95$. Its power therefore cannot be
recovered from a zero-effect design at all, which is why the main text measures
detection directly rather than inferring it.

\begin{table}[H]
\centering\footnotesize
\renewcommand{\arraystretch}{0.86}
\caption{Zero-effect simulations, $N=T=40$, 200 replications. The true ATT is zero, so Bias is the mean estimate. Cov is coverage of the nominal 95\% interval, Len its mean length, and IS the mean interval score $(u-l)+(2/\alpha)\{(l-y)1(y<l)+(y-u)1(y>u)\}$, which penalizes width and non-coverage together; lower is better, and IS equals Len whenever coverage is one. Proposed intervals use the oracle radius. Honest DiD returns a robust set, not a point estimate.}
\label{tab:sim}
\begin{tabular}{llccccc}
\toprule
\multicolumn{7}{l}{\textit{Panel A: correct rank ($R=2$)}} \\
Factors & Method & Cov & Len & IS & Bias & RMSE \\
\midrule
\multirow{5}{*}{$\kappa=1$} & Proposed & 1.00 & 6.51 & 6.51 & 0.69 & 0.78 \\
 & GSC & 0.47 & 1.49 & 8.90 & 0.77 & 0.86 \\
 & C\&S DiD & 0.34 & 2.09 & 17.39 & 1.26 & 1.38 \\
 & Honest DiD ($\bar M=0.5$) & 1.00 & 6.23 & 6.23 & -- & -- \\
\cmidrule(l){1-7}
\multirow{5}{*}{$\kappa=0.25$} & Proposed & 1.00 & 6.01 & 6.01 & 0.62 & 0.64 \\
 & GSC & 0.03 & 0.65 & 16.70 & 0.72 & 0.74 \\
 & C\&S DiD & 0.83 & 2.02 & 3.34 & 0.34 & 0.65 \\
 & Honest DiD ($\bar M=0.5$) & 1.00 & 6.32 & 6.32 & -- & -- \\
\midrule
\multicolumn{7}{l}{\textit{Panel B: rank underfitted ($R=1$, truth $R=2$)}} \\
Factors & Method & Cov & Len & IS & Bias & RMSE \\
\midrule
\multirow{5}{*}{$\kappa=1$} & Proposed & 1.00 & 9.82 & 9.82 & 2.38 & 2.42 \\
 & GSC & 0.00 & 1.78 & 78.72 & 2.81 & 2.86 \\
 & C\&S DiD & 0.34 & 2.09 & 17.39 & 1.26 & 1.38 \\
 & Honest DiD ($\bar M=0.5$) & 1.00 & 6.23 & 6.23 & -- & -- \\
\bottomrule
\end{tabular}
\end{table}

\begin{figure}[H]
\centering
\includegraphics[width=0.9\linewidth]{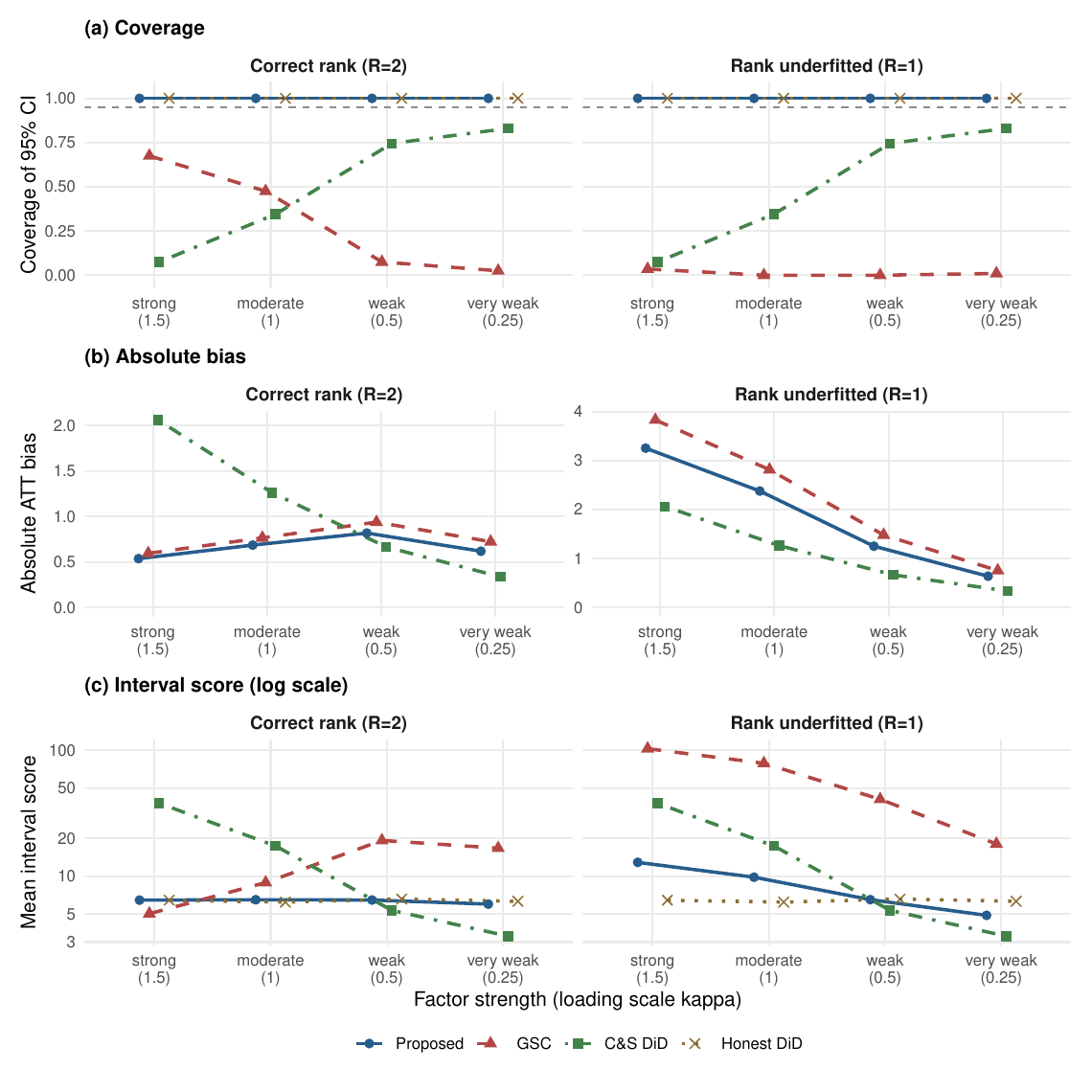}
\caption{Coverage, absolute ATT bias and mean interval score by factor
strength, true ATT zero. The left column fits the true rank, the right fits one
factor when two are present. The dashed line is nominal $95\%$. The two bias
panels have separate vertical scales, and the interval-score panels are
logarithmic, the range running from $3.3$ to $102.9$. Honest DiD returns a set
rather than a point estimate, so it appears in the coverage and score rows but
not in the bias row.}
\label{fig:sim}
\end{figure}

\section*{S7. Aggregation sensitivity of the parallel-trends comparators}

The two parallel-trends benchmarks impose the same identifying assumption but
differ sharply in coverage, and the difference is aggregation rather than
identification. The Callaway--Sant'Anna simple ATT averages every
post-treatment period, including those far from adoption where the linear
confounding trend has accumulated most. Honest DiD at $\bar M=0$ applies the
same assumption to an event-study aggregation over a window of event times
$-4$ to $3$. Table~\ref{tab:aggregation} reports both over the same $200$
draws.

\begin{table}[H]
\centering\small
\caption{Coverage of the two parallel-trends benchmarks on identical draws,
$200$ replications, true ATT zero. Both impose parallel trends; they differ
only in how post-treatment periods are aggregated.}
\label{tab:aggregation}
\begin{tabular}{lcccc}
\toprule
Factor strength & $\kappa=1.5$ & $\kappa=1$ & $\kappa=0.5$ & $\kappa=0.25$ \\
\midrule
C\&S simple ATT & 0.08 & 0.34 & 0.74 & 0.83 \\
Honest DiD, $\bar M=0$ & 0.80 & 0.88 & 0.90 & 0.90 \\
\bottomrule
\end{tabular}
\end{table}

Reports that factor-based methods dominate parallel-trends methods in designs
of this kind should therefore state which aggregation is being compared.

\section*{S8. Algorithm and reported specifications}

\begin{enumerate}
\item Construct $Y$, $D$ and covariates; define $\mathcal O=\{D_{ti}=0\}$.
\item Choose $R$ and estimate $(\beta,\alpha,\delta,L)$ on $\mathcal O$ only.
\item Complete $\widehat\Gamma$ and form untreated residuals and
      $\widehat\sigma$.
\item Form $W=D/|\mathcal T|$ and solve the convex minimax program for
      $\widehat A$ at a chosen radius.
\item Report $\widehat\tau=\langle W,Y-\widehat\Gamma\rangle
      -\langle\widehat A,Y-\widehat\Gamma\rangle$ with its bias allowance and
      the interval of Section~S4.
\item Repeat over $C$ and report the breakdown radius.
\end{enumerate}

The simulations fix $N=T=40$, two-way effects, no covariates, and five
three-unit adoption cohorts at dates $20$--$34$; the correctly specified
design fits $R=2$ and the misspecified design fits $R=1$ against rank-two
truth, with all else held fixed. The nonzero-effect designs of Section 3 use
$500$ replications and the zero-effect designs reported in this Supplement use
$200$. The EDR application uses
two-way effects, mail-in and motor-voter covariates and $R=2$. Conventional GSC
uses the same ranks, covariates, untreated observations and IFE routine.
Callaway--Sant'Anna and Rambachan--Roth honest DiD are estimated separately as
parallel-trends benchmarks; for honest DiD the event-study covariance is
recovered from the influence function returned by the aggregation step, and
event times are indexed by election number rather than calendar year.

\end{document}